\documentclass[11pt,a4paper]{article}
\usepackage{indentfirst,mathrsfs}
\usepackage{amsfonts,amsmath,amssymb,amsthm}
\usepackage{latexsym,amscd}
\usepackage{amsbsy}
\usepackage{color}
\usepackage{cases}
\newtheorem{lem}{Lemma}

\newtheorem{conj}{Conjecture}

\newcommand{\ftwon}{{\mathbb F}_{2^n}}
\newcommand{\ftwom}{{\mathbb F}_{2^m}}
\newcommand{\ftwo}{{\mathbb F}_{2}}

\begin{document}
\title{A conjecture on permutation trinomials over finite fields of characteristic two}
\author{Nian Li
\thanks{Faculty of Mathematics and Statistics, Hubei Key Laboratory of Applied Mathematics, Hubei
University, Wuhan, 430062, China. Email: nian.li@hubu.edu.cn;huqiaoyu@stu.hubu.edu.cn}
\and Qiaoyu Hu
}
\date{}
\maketitle

\begin{abstract}
 In this paper, by analyzing the quadratic factors of an $11$-th degree polynomial over the finite field $\ftwon$, a conjecture on permutation trinomials over $\ftwon[x]$ proposed very recently by Deng and Zheng  is settled, where $n=2m$ and $m$ is a positive integer with $\gcd(m,5)=1$.
 % (Cryptography and Communication, DOI https://doi.org/10.1007/s12095-018-0284-7)
\end{abstract}

\section{Introduction}

 Let $q$ be a prime power and $\mathbb{F}_{q}$ denote the finite field with $q$ elements. A polynomial $f(x)$ over $\mathbb{F}_{q}$ is called a permutation polynomial if the induced mapping $f:c\mapsto f(c)$ from $\mathbb{F}_{q}$ to itself is a bijection \cite{lidl1997}. Permutation polynomials have been studied for several decades and have important applications in a wide range of areas such as coding theory \cite{DT2013,Y2007}, combinatorial designs \cite{DY2006} and cryptography \cite{RS1987,SH1988}.

The construction of permutation polynomials with a simple algebraic form is an interesting research problem and it has already attracted researchers' much attention in recent years. By using certain techniques in dealing with equations or polynomials over finite fields, a number of permutation polynomials with a simple form have been obtained, the reader is referred to \cite{cding,gupta,hou3,hou4,H2015,likangquan,linian,MZFG,Tu-zeng-h,Tu-zeng-j,dwu,zha} and the references therein. Motivated by the observation that more than half of the known permutation binomials and trinomials were constructed from Niho exponents, Li and Helleseth \cite{linian} aimed to investigate permutation trinomials over $\ftwon[x]$ of the form
\begin{equation}\label{e1}
f(x)=x+x^{s(2^m-1)+1}+x^{t(2^m-1)+1},
\end{equation}
where $n=2m$, and $s$, $t$ are two integers, and consequently, four classes of permutation trinomials over $\ftwon[x]$ with the form \eqref{e1} were obtained in \cite{linian} based on some subtle manipulation of solving equations with low degree over finite fields, and another two classes of such permutations were presented in \cite{linian2} by virtue of the property of linear fractional polynomials over finite fields. Meanwhile, some similar and more general results on permutation trinomials over $\ftwon[x]$ were also obtained in \cite{gupta,Li-QLF}. For the permutation polynomials from Niho exponents, the reader is referred to \cite{Tu-zeng,Tu-zh,Tu-zlh} for some recent results and to a survey paper \cite{li-zeng}. Very recently, followed the work of \cite{linian}, by some delicate operation of solving equations with low degrees over finite fields, Deng and Zheng \cite{DZ2018} presented two more classes of permutation trinomials over $\ftwon[x]$ of the form \eqref{e1}, and proposed a conjecture on such a kind of permutation trinomials based on computer experiments. This paper is devoted to settle the conjecture proposed by Deng and Zheng in \cite{DZ2018}.

The remainder of this paper is organized as follows. Section \ref{sec-2} gives some notations and the conjecture proposed in \cite{DZ2018}.  Section \ref{sec-3} proves the conjecture by analyzing the quadratic factors of an $11$-th degree polynomial over the finite field $\ftwon$, and Section \ref{sec-4} concludes this paper.

\section{A conjecture on permutation trinomials of the form \eqref{e1}}\label{sec-2}

A criterion for a polynomial of the form \eqref{e1} to be a permutation polynomial had been characterized by the following lemma which was proved by Park and Lee \cite{Park}, Wang \cite{Wang07} and Zieve \cite{Zieve}.

\noindent
\begin{lem}\label{lem1} \emph{(\cite{Park,Wang07,Zieve})}
Let $q$ be a prime power and $h(x)\in\mathbb{F}_{q}[x]$. If $d, s, r>0$ such that $q-1=ds$, then $x^rh(x^s)$ permutes $\mathbb{F}_{q}$ if and only if

\rm{(1)} $\gcd\left(r,s\right)=1$;

\rm{(2)} $x^rh(x)^s$ permutes the set of the $d$-th roots of unity in $\mathbb{F}_{q}$.

\end{lem}
 
From now on, let $n=2m$ be a positive integer and denote the $(2^m+1)$-th roots of unity in $\ftwon$, i.e., the unit circle of $\ftwon$ by 
\[{\mu}_{2^m+1}=\{x\in \ftwon| x^{2^m+1}=1\}.\]

\begin{conj}\label{conj1} \emph{(\cite{DZ2018})}
Let $\gcd(m,5)=1$ and $(s,t)=(\frac{4}{11},\frac{10}{11})$, then $f(x)$ defined by \eqref{e1} is a permutation polynomial over $\mathbb{F}_{2^{2m}}$.
\end{conj}

Note that $\gcd(11,2^m+1)=1$ if $\gcd(m,5)=1$. Then, according to Lemma \ref{lem1}, to prove Conjecture \ref{conj1}, it suffices to show $x(1+x^{\frac{4}{11}}+x^{\frac{10}{11}})^{2^m-1}$, or equivalently 
\begin{equation}\label{e2}
x^{11}(1+x^4+x^{10})^{2^m-1}
 \end{equation}
permutes the unit circle of $\mathbb{F}_{2^{2m}}$. Observe that $x^5+x^2+1$ is irreducible of $\mathbb{F}_{2}[x]$ which implies that $x^5+x^2+1=0$ has solutions in $\mathbb{F}_{2^m}$ if and only if $m\equiv0\pmod 5$. Thus, by $\gcd(m,5)=1$, one gets $x^{10}+x^{4}+1=(x^{5}+x^{2}+1)^2\neq0$ for any $x\in \mu_{2^m+1}\subset \mathbb{F}_{2^{2m}}$. It is therefore (\ref{e2}) can be written as
\begin{equation}\label{e3}
x^{11}\cdot\frac{1+x^{-4}+x^{-10}}{1+x^4+x^{10}}=\frac{x^{11}+x^{7}+x}{x^{10}+x^4+1}.
 \end{equation}\\
Then, to prove Conjecture \ref{conj1}, it suffices to show that (\ref{e3}) permutes the unit circle of $\mathbb{F}_{2^{2m}}$ if $\gcd(m,5)=1$, i.e.,
\begin{equation*}
\frac{x^{11}+x^{7}+x}{x^{10}+x^4+1}=t
 \end{equation*}
has a unique solution in $\mu_{2^m+1}$ for any $t\in \mu_{2^m+1}$ if $\gcd(m,5)=1$, which is equivalent to proving that the equation
\begin{equation}\label{e4}
x^{11}+tx^{10}+x^{7}+tx^{4}+x+t=0
 \end{equation}
has at most one solution in $\mu_{2^m+1}$ for any $t\in \mu_{2^m+1}$ if $\gcd(m,5)=1$.

\section{Proof of Conjecture \ref{conj1} }\label{sec-3}

This section presents the proof of Conjecture \ref{conj1}.

\begin{lem}\label{lem2} %\emph{(\cite{5})}
Let $F(x)=x^{11}+tx^{10}+x^{7}+tx^{4}+x+t$, where $t\in \mu_{2^m+1}$. If $x^2+ax+b$, where $ab\neq0$, is a factor of $F(x)$, then $a$,$b$ must satisfy one of the following conditions:

\rm{(1)} $ab^3+b^2+b+a^2=0$ and $b^6+(a^4+1)b^4+b^3+a^2=0$;

\rm{(2)} $b^2+b^3+a^2b^2+a=0$ and $a^2b^6+b^5+b^4+b^2+a^4=0$;

\rm{(3)} $a^{10}+(b^4+b^2+1)a^6+(b^5+b)a^4+(b^7+b)a^2+b^{10}+b^{8}+b^{7}+
b^{5}+b^{3}+b^{2}+1=0$.
\end{lem}

\begin{proof}
  Assume that $F(x)$ can be factorized as
\begin{equation*}
F(x)=(x^2+ax+b)(x^9+\sum\limits_{i=1}^{9}c_ix^{9-i}).
\end{equation*}
Expanding the right hand side of $F(x)$ and comparing the coefficients of $x^{11-i}$ where $i=1,2,3,4,5$ gives
\begin{eqnarray*}
c_{1}&=&a+t,\\
c_{2}&=&b+a^2+at,\\
c_{3}&=&a^3+a^2t+bt,\\
c_{4}&=&1+b^2+a^2b+a^4+a^3t,\\
c_{5}&=&a+ab^2+a^5+a^2bt+b^2t+a^4t,
\end{eqnarray*}
and comparing the coefficients of $x^{i}$ for $i=0,1,2,3,4,5$ gives
\begin{eqnarray*}
c_{9}&=&\frac{t}{b},\\
c_{8}&=&\frac{b+at}{b^2},\\
c_{7}&=&\frac{bt+ab+a^2t}{b^3},\\
c_{6}&=&\frac{b^2+a^2b+a^3t}{b^4},\\
c_{5}&=&\frac{b^4t+b^2t+a^2bt+a^4t+a^3b}{b^5},\\
c_{4}&=&\frac{b^3+a^2b^2+ab^4t+ab^2t+a^5t+a^4b}{b^6}.
\end{eqnarray*}
Then, according to the values of $c_4$ and $c_5$, one gets
\[
\left\{\begin{aligned}
&(1+b^2+a^2b+a^4+a^3t)b^6=b^3+a^2b^2+a^4b+(ab^4+ab^2+a^5)t, \\
&(a+ab^2+a^5+a^2bt+b^2t+a^4t)b^5=a^3b+(b^4+b^2+a^2b+a^4)t,  \\
\end{aligned}\right.
\]
i.e.,
\begin{numcases}{}
   (ab^4+ab^2+a^5+a^3b^6)t=b^6+b^8+a^2b^7+a^4b^6+b^3+a^2b^2+a^4b, \label{e5}\\
   (a^2b^6+b^7+a^4b^5+b^4+b^2+a^2b+a^4)t=ab^5+ab^7+a^5b^5+a^3b.\label{e6}
\end{numcases}

In the following  we shall consider three cases to prove Lemma \ref{lem2}.

\textbf{Case 1.} If $ab^4+ab^2+a^5+a^3b^6=0$, i.e., $(b^2+b+a^2+ab^3)^2=0$ since $a\ne0$. 
Then by (\ref{e5}), one obtains that $b^6+b^8+a^2b^7+a^4b^6+b^3+a^2b^2+a^4b=0$, i.e., $b^5+b^7+a^2b^6+a^4b^5+b^2+a^2b+a^4=0$. Replacing $a^2b^6+a^4$ by $b^4+b^2$ gives 
\[b^5+b^7+a^4b^5+b^2+a^2b+b^4+b^2=0.\] 
Thus, in this case we have
\begin{numcases}{}
ab^3+b^2+b+a^2=0,  \label{e7}\\
b^6+(a^4+1)b^4+b^3+a^2=0.\label{e8}
\end{numcases}

\textbf{Case 2.} If $a^2b^6+b^7+a^4b^5+b^4+b^2+a^2b+a^4=0$, by (\ref{e6}) one gets $ab^5+ab^7+a^5b^5+a^3b=0$, which implies  $ab(b^2+b^3+a^2b^2+a)^2=0$, and one then has
\begin{numcases}{}
b^2+b^3+a^2b^2+a=0, \label{e9}\\
a^2b^6+b^5+b^4+b^2+a^4=0\label{e10}
\end{numcases}
due to $b^7+a^4b^5+a^2b=b(b^3+a^2b^2+a)^2=b^5$.

\textbf{Case 3.} If $ab^4+ab^2+a^5+a^3b^6\ne0$ and $a^2b^6+b^7+a^4b^5+b^4+b^2+a^2b+a^4\ne0$. Then, by (\ref{e5}) and (\ref{e6}), one gets $$\frac{b^6+b^8+a^2b^7+a^4b^6+b^3+a^2b^2+a^4b}{ab^4+ab^2+a^5+a^3b^6}
=\frac{ab^5+ab^7+a^5b^5+a^3b}{a^2b^6+b^7+a^4b^5+b^4+b^2+a^2b+a^4},$$
which can be written as $$a^{10}+(b^4+b^2+1)a^6+(b^5+b)a^4+(b^7+b)a^2+b^{10}+b^8+b^7+b^5+b^3+b^2+1=0$$
by a detailed calculation. This completes the proof.
\end{proof}

To prove Conjecture 1, we need to show that $F(x)$ in Lemma \ref{lem2} cannot have two solutions in $\mu_{2^m+1}$ for any $t\in\mu_{2^m+1}$, i.e., $F(x)$ cannot have a quadratic factor $x^2+ax+b=0$ satisfying $x^2+ax+b=0$ have two solutions in $\mu_{2^m+1}$. Observe that if $x_1, x_2 \in \mu_{2^m+1}$ are two solutions to $x^2+ax+b=0$, then
\[x_1+x_2=a, x_1x_2=b.\]
Moreover, one has
\[x_1^{2^m}+x_2^{2^m}=a^{2^m}=\frac{1}{x_1}+\frac{1}{x_2}=\frac{x_1+x_2}{x_1x_2}=\frac{a}{b},\]
i.e., $a^{2^m}b=a$. This implies that if $x^2+ax+b=0$ is a factor of $F(x)$ satisfying $x^2+ax+b=0$ has two solutions in $\mu_{2^m+1}$, then it must have $a^{2^m}b=a$. Actually this fact has been found in \cite{Tu-zlh} and the number of solutions $x$ in $\mu_{2^m+1}$ to $x^2+ax+b=0$ has also been characterized there. We provide the proof of the relation $a^{2^m}b=a$ here to make the paper self-contained.

Due to this fact, we further consider the conditions in Lemma \rm{\ref{lem2}}.

\begin{lem}\label{lem3}
Let $F(x)=x^{11}+tx^{10}+x^7+tx^4+x+t$, where $t\in\mu_{2^m+1}$. If $x^2+ax+b$, where $ab\ne0$ and $a^{2^m}b=a$, is a quadratic factor of $F(x)$, then $a,b$ must satisfy
\begin{equation}\label{e11}
v^5+(u^4+u^2)v^3+v^2+(u^6+u^4)v+u^{10}+u^4+1=0,
\end{equation}
where $u=a^{-1}+a^{-2^m}$ and $v=a^{-1}\cdot a^{-2^m}$.
\end{lem}

\begin{proof}
According to Lemma \ref{lem2}, we can discuss the three cases in Lemma \ref{lem2} as follows:
\begin{enumerate}
  \item [(1)]
   Taking $2^m$-th power on both sides of (\ref{e7}) gives
   \begin{equation}\label{e12}
     a^{2^m}b^{3\cdot2^m}+b^{2\cdot2^m}+b^{2^m}+a^{2\cdot2^m}=0.
   \end{equation}
  On the other hand, by $a^{2^m}b=a$ and $ab^{2^m}=a^{2^m}$, one gets $a^{2^m}=a/b$, $b^{2^m}=a^{2^m}/a=1/b$ and then
\begin{eqnarray*}
a^{2^m}b^{3\cdot2^m}=\frac{a}{b}\cdot\frac{1}{b^3}=\frac{a}{b^4},b^{2\cdot2^m}=\frac{1}{b^2},b^{2^m}=\frac{1}{b},a^{2\cdot2^m}=\frac{a^2}{b^2}.
\end{eqnarray*}

Then, (\ref{e12}) can be written as
   \begin{equation}\label{e13}
   \frac{a}{b^4}+\frac{1}{b^2}+\frac{1}{b}+\frac{a^2}{b^2}=\frac{1}{b^4}(a+b^2+b^3+a^2b^2)=0.
   \end{equation}
 Combining (\ref{e7}) and (\ref{e13}) gives
 \[(ab^3+a^2)b=a+a^2b^2,\]
 i.e., \[ab(b+1)=(b+1)^4.\]

 If $b=1$, then (\ref{e7}) leads to $a=1$ since $a\neq0$, a contradiction to (\ref{e8}). Thus, one gets $b\neq1$ and then $a=(b+1)^3b^{-1}$. Substituting it into (\ref{e7}), one obtains
\[b^2(b+1)^3+b^2+b+\frac{(b+1)^6}{b^2}=0,\]
which can be reduced to $b^6+b^5+b^3+b+1=(b^2+b+1)^3=0$, i.e., $b^2+b+1=0$. Then, $b^3=1$ and $a=(b+1)^3b^{-1}=b^6b^{-1}=b^{-1}=b^2$.
Thus, (\ref{e8}) is reduced to
\[1+(b^8+1)b^4+1+b^4=0,\]
i.e., $0=(b^2+1)b+b=b^3$, a contradiction to $b^2+b+1=0$. Then, the first case in Lemma \ref{lem2} cannot happen if $a^{2^m}b=a$.

  \item [(2)]
  Taking $2^m$-th power on both sides of (\ref{e9}) gives
    \begin{equation}\label{e14}
     b^{2\cdot2^m}+b^{3\cdot2^m}+a^{2\cdot2^m}b^{2\cdot2^m}+a^{2^m}=0.
   \end{equation}
   Again by $a^{2^m}=a/b$ and $b^{2^m}=1/b$, one has
\begin{eqnarray*}
b^{2\cdot2^m}=\frac{1}{b^2},b^{3\cdot2^m}=\frac{1}{b^3},a^{2\cdot2^m}b^{2\cdot2^m}=\frac{a^2}{b^4},a^{2\cdot2^m}=\frac{a^2}{b^2}. 
%b^{2\cdot2^m}&=&\frac{(ab^{2^m})^2}{a^2}=\frac{a^{2}}{a^2b^2}=\frac{1}{b^2},\\
%b^{3\cdot2^m}&=&\frac{(ab^{2^m})^3}{a^3}=\frac{a^3}{a^3b^3}=\frac{1}{b^3},\\ a^{2\cdot2^m}b^{2\cdot2^m}&=&\frac{(a^{2^m}b)^2(ab^{2^m})^2}{b^2a^2}=\frac{a^2a^2}{b^4a^2}=\frac{a^2}{b^4},\\ a^{2\cdot2^m}&=&\frac{(a^{2^m}b)^2}{b^2}=\frac{a^2}{b^2}.
\end{eqnarray*}
Then, (\ref{e14}) can be written as
\[\frac{1}{b^2}+\frac{1}{b^3}+\frac{a^2}{b^4}+\frac{a^2}{b^2}=\frac{1}{b^4}(b^2+b+a^2+a^2b^2)=0.\]
This together with (\ref{e9}) implies that
\begin{equation}\label{e-a1}
(a^2+a^2b^2)b=a^2b^2+a,\;\; {\rm i.e.,}\;\; a=\frac{1}{b(b^2+b+1)}
\end{equation} 
since $ab\neq0$. Substituting it into (\ref{e9}) gives
\[b^2+b^3+\frac{1}{(b^2+b+1)^2}+\frac{1}{b(b^2+b+1)}=0,\]
which is equivalent to
 \begin{equation}\label{e-b}
 (b+1)(b^7+b^5+b^3+b+1)=0. 
 \end{equation} 
  
  Since $a^{2^m}=a/b$ and $b^{2^m}=1/b$, then by \eqref{e-a1} one has
  \[\frac{a}{b}=a^{2^m}=\frac{1}{b^{2^m}(b^{2\cdot2^m}+b^{2^m}+1)}=\frac{1}{b^{-1}(b^{-2}+b^{-1}+1)}=\frac{b^3}{b^2+b+1},\]
 which together with \eqref{e-a1} implies that $b^5=1$.
  If $b=1$, then $a=\frac{1}{b(b^2+b+1)}=1$, a contradiction to (\ref{e10}). If $b\ne 1$, then by $b^5=1$ and \eqref{e-b} one gets $b^7+b^5+b^3+b+1=b^2+b^3+b=b(b^2+b+1)=0$, a contradiction to \eqref{e-a1}. 
  Thus, the second case in Lemma \ref{lem2} cannot happen either if $a^{2^m}b=a$.

  \item [(3)] By $a^{2^m}b=a$, one gets $a^{2^m}=a/b$, $a^{-2^m}=b/a$, which implies 
  \[a^{-1}+a^{-2^m}=\frac{b+1}{a}, \;\;a^{-1}\cdot a^{-2^m}=\frac{b}{a^2}.\]
Let $u=a^{-1}+a^{-2^m}$ and $v=a^{-1}\cdot a^{-2^m}$, then we have
\begin{itemize}
   \item [1)] $\frac{a^6(b^4+b^2+1)}{a^{10}}=\frac{b^4+b^2+1}{a^4}=(\frac{b+1}{a})^4+(\frac{b}{a^2})^2=u^4+v^2$;
   \item [2)] $\frac{a^4(b^5+b)}{a^{10}}=(\frac{b+1}{a})^4\cdot\frac{b}{a^2}=u^4v$;
   \item [3)] $\frac{a^2(b^7+b)}{a^{10}}=\frac{b(b^6+1)}{a^8}=\frac{b[(b+1)^6+(b^4+b^2)]}{a^2\cdot a^6}=vu^6+v^3u^2$;
   \item [4)] $\frac{b^{10}+b^8+b^7+b^5+b^3+b^2+1}{a^{10}}=\frac{(b+1)^{10}+b^7+b^5+b^3}{a^{10}}=u^{10}+v^3u^4+v^5$.
\end{itemize}
\end{enumerate}
Thus, by Lemma \ref{lem2} (3), one obtains 
$$v^5+(u^4+u^2)v^3+v^2+(u^6+u^4)v+u^{10}+u^4+1=0.$$
 This completes the proof.
\end{proof}

Notice that if (\ref{e4}) has two or more solutions in $\mu_{2^m+1}$ for some $t$ in $\mu_{2^m+1}$, then $F(x)$ must have a quadratic factor, say $x^2+ax+b$. If so, by Lemma \ref{lem3}, the coefficients $a$, $b$ must satisfy the relation (\ref{e11}). Then, to prove Conjecture \ref{conj1}, we only need to show that (\ref{e11}) has no solution in $\mathbb{F}_{2^m}$ if $\gcd(m,5)=1$,

To this end, define
\[G(x,y)=x^5+(y^2+y)x^3+x^2+(y^3+y^2)x+y^5+y^2+1.\]

Notice that $x^5+x^2+1$ is irreducible over $\mathbb F_2[x]$ and $\gcd(m,5)=1$, then all the roots of $x^5+x^2+1=0$ are conjugate over $\ftwo$ and lie in $\mathbb F_{2^{5m}}$ \cite[Thm. 2.14]{lidl1997}.
Define $$H=\{x\in \mathbb F_{2^{5m}}|x^5+x^2+1=0\}.$$

\begin{lem}\label{lem4}
The polynomial $G(x,y)$ can be factorized over $\mathbb F_{2^{5m}}[x,y]$ as follows:
\[G(x,y)=\prod_{\theta\in H}(x+\theta^{-1}y+\theta).\]
\end{lem}

\begin{proof}
Since all the roots of $x^5+x^2+1=0$ are distinct and lie in $H$, thus to complete the proof, we only need to show that $G(x,y)=0$ if $y=\theta x+\theta^2$ for any $\theta\in H$. By a direct computation, if $y=\theta x+\theta^2$, then one gets
\begin{eqnarray*}
  (y^2+y)x^3  &=& \theta^2 x^5+\theta x^4+(\theta^4+\theta^2)x^3, \\
  (y^3+y^2)x &=& \theta^3x^4+(\theta^4+\theta^2)x^3+\theta^5x^2+(\theta^6+\theta^4)x, \\
  y^5+y^2+1 &=&  \theta^5x^5+\theta^6x^4+\theta^2x^2+\theta^9x+\theta^{10}+\theta^4+1,
\end{eqnarray*}
which  implies that 
\begin{eqnarray*}
  G(x,\theta x+\theta^2) &=& (\theta^5+\theta^2+1)x^5+(\theta^6+\theta^3+\theta)x^4+ (\theta^5+\theta^2+1)x^2\\
    &&+(\theta^9+\theta^6+\theta^4)x +(\theta^5+\theta^2+1)^2. 
\end{eqnarray*}
Note that $\theta\in H$. Then, $\theta^5+\theta^2+1=0$, $\theta^6+\theta^3+\theta=\theta(\theta^5+\theta^2+1)=0$ and $\theta^9+\theta^6+\theta^4=\theta^4(\theta^5+\theta^2+1)=0$, i.e., $ G(x,\theta x+\theta^2)=0$ for any $x$ and any $\theta\in H$.
This implies the proof.
\end{proof}

\begin{lem}\label{lem5}
$G(x,y)=0$ has no solution $(x,y)\in \mathbb F_{2^{m}}\times\mathbb F_{2^{m}}$ if $\gcd(m,5)=1$.
\end{lem}
\begin{proof}
 According to Lemma \ref{lem4}, it suffices to show that
 \begin{equation}\label{e15}
x+\theta^{-1}y+\theta=0
 \end{equation}
cannot hold for $x,y\in \mathbb F_{2^{m}}$, $\theta\in \mathbb F_{2^{5m}}$ with $\theta^5+\theta^2+1=0$.

Suppose that $x,y\in \mathbb F_{2^{m}}$, then taking $2^m$-th power on both sides of (\ref{e15}) gives
 \begin{equation}\label{e16} 
x+\theta^{-2^m}y+\theta^{2^m}=0.
 \end{equation}
Notice that $\theta^{-1}\ne \theta^{-2^m}$. Otherwise we have $\theta\in\ftwom$, a contradiction to the facts that $\gcd(m,5)=1$ and $\theta\in \mathbb F_{2^{5m}}$ satisfying $\theta^5+\theta^2+1=0$. Then, by (\ref{e15}) and (\ref{e16}), one can obtain that
\[
\left\{\begin{aligned}
&x=\theta+\theta^{2^m},\\
&y=\frac{\theta+\theta^{2^m}}{\theta^{-1}+\theta^{-2^m}}.\\
\end{aligned}\right.
\]\\
Again by $x\in \mathbb F_{2^{m}}$, i.e., $x^{2^m}=x$, one has
\[ (\theta+\theta^{2^m})^{2^m}=\theta+\theta^{2^m},\]
then $\theta\in \mathbb F_{2^{2m}}$, which leads to
 \[\theta\in \mathbb F_{2^{2m}}\cap\mathbb F_{2^{5m}}=\ftwom,\]
 a contradiction to $\theta^5+\theta^2+1=0$ and $\gcd(m,5)=1$. This completes the proof.
\end{proof}

\section{Conclusion}\label{sec-4}

In this paper, by analyzing the possible quadratic factors of an $11$-th degree polynomial over the finite field $\ftwon$, a conjecture on permutation trinomials over $\ftwon[x]$ proposed by Deng and Zheng in \cite{DZ2018}  was settled.

\end{document}